\RequirePackage{fix-cm}
\documentclass{article}
\usepackage{graphicx}
\usepackage{mathptmx}      
\usepackage{natbib}
\usepackage[stdsize,thm]{botex}
\usepackage{amsmath}
\usepackage{verbatim}
\usepackage{prettyref}
\usepackage{subfigure}
\usepackage{tikz}
\usepackage[colorlinks=true,urlcolor=blue,citecolor=blue]{hyperref}
\newcommand{\pdiff}{\ensuremath{\mathcal{E}}}

\renewcommand{\vec}[1]{\mathbf{#1}}
\newcommand{\al}{\alpha}
\newcommand{\tp}{\ensuremath{^\top}}
\begin{document}

\title{A revised model of fluid transport optimization\\ in \emph{Physarum polycephalum}}
\author{Vincenzo Bonifaci\thanks{Istituto di Analisi dei Sistemi ed Informatica, Consiglio Nazionale delle Ricerche, Rome, Italy. Email: \url{vincenzo.bonifaci@iasi.cnr.it}}}
\date{}

\maketitle

\begin{abstract}
Optimization of fluid transport in the slime mold \emph{Physarum} \emph{polycephalum} has been the subject of several modeling efforts in recent literature. Existing models assume that the tube adaptation mechanism in \emph{P. polycephalum}'s tubular network is controlled by the sheer amount of fluid flow through the tubes. We put forward the hypothesis that the controlling variable may instead be the flow's pressure gradient along the tube. We carry out the stability analysis of such a revised mathematical model for a parallel-edge network, proving that the revised model supports the global flow-optimizing behavior of the slime mold for a substantially wider class of response functions compared to previous models. 
Simulations also suggest that the same conclusion may be valid for arbitrary network topologies. 
\end{abstract}

\section{Introduction}
\emph{Physarum polycephalum} is an amoeboid slime mold that exhibits remarkable information processing capabilities. In controlled experiments, the slime mold's abilities have been leveraged to determine the shortest path between two locations in a network  \citep{Nakagaki:2000,Tero:2006} and, more generally, to adaptively form efficient transport networks \citep{Tero:2010}. The question remains, however, of detailing and analyzing the underlying mechanisms and goals of such an optimization process, which have been only partially explored \citep{Tero:2007,Miyaji:2007,Tero:2010,Ito:2011,Bonifaci:2012:b,Ma:2013,Bonifaci:2013}.  

\emph{P.~polycephalum}, like other Myxomycetes, has a somewhat complicated life cycle consisting of several stages \citep{Stephenson:2000}. During its mature plasmodium stage it forms a single, giant multinucleate acellular structure. This acellular structure takes the form of a tubular vein network, through which protoplasm is periodically transported, driven by the gradient of hydrostatic pressure. 

In a remarkable experiment by \citet*{Nakagaki:2000}, the plasmodium has been placed on a preexisting artificial network structure and two food sources (oat flakes) have been laid over two nodes $s_0$, $s_1$ of the network. \emph{P.~polycephalum} reacted by adapting its shape dynamically, by controlling the width of the tubular structures forming its veins, based on feedback from the protoplasmic flow. Gradually, several branches of the vein network collapsed, eventually leaving the mass of the slime mold only along the shortest path between the two food sources $s_0$ and $s_1$, thus exhibiting evidence of fluid transport optimization. 

\citet*{Tero:2007} were the first to propose a mathematical model for the transport optimization dynamics of \emph{P.~polycephalum}. For a critical value of the model parameters, the model's dynamics indeed provably converge to the shortest path between the two food-source terminals of the underlying network. Such a convergence to the shortest path has been analytically proven by \citet*{Tero:2007} for a ring-shaped network, by \citet*{Miyaji:2007} for a Wheatstone bridge-shaped network, and finally by \citet*{Bonifaci:2012:b} for arbitrary complex networks. 
However, the global convergence of the dynamics to the shortest path fails for other values of the model's power-law exponent, called $\mu$ by \citet{Tero:2007}, as well as for nonzero values of the parameter $\alpha$ controlling saturation of the tubular dynamics. In this article we propose a revised model that, while enlarging the class of admissible response functions, has the property of converging to the shortest path in the network from any initial condition, independently of the specific details of the response dynamics or of the values of its parameters. We prove this analytically for a parallel-edge network, and test our claim by simulation on more complex topologies. 

Fluid transport in \emph{P.~polycephalum} is based on the positive feedback mechanism between the width of the veins and the fluid flow. On a longer time scale, it is driven by network peristalsis, which reverses the flow velocities periodically \citep{Alim:2013,Baumgarten:2013}. It has been found that the time between reversals is much larger than the time required for the veins to adapt their widths under a steady flow. In this work we focus on the vein adaptation dynamics, without attempting to model the peristaltic oscillations. The reader is referred to \citet{Kobayashi:2006} for one possible model of the rhythmic protoplasmic movement.  
Other details on the physical underpinnings of Physarum's tubular dynamics are discussed by \citet{Tero:2005}. Finally, the slime mold network dynamics have also been considered from a combinatorial optimization perspective in the discrete algorithms literature \citep{Becchetti:2013,Straszak:2016:a}.

\subsection{Outline of the paper} 
We recall the Tero-Kobayashi-Nakagaki network adaptation model in Section \ref{sec:tkn}, and present a revised model in Section \ref{sec:revised-model}. A stability analysis of the new model's equilibria is carried out in Section \ref{sec:parallel} for networks consisting of parallel edges; it is proved that the unique stable fixed point corresponds to the shortest path in the graph, independently of the power-law exponent and of the saturation parameter, or, indeed, of the particular shape of the tubular response functions. Moreover, any nontrivial trajectory approaches such a fixed point. 
In Section \ref{sec:simul} we report on simulations for more complex network topologies; these suggest that the globally optimal behavior of the dynamics may hold in arbitrary topologies, again independently of the particular shape of the tubular response functions. We close by summarizing and discussing our findings in Section \ref{sec:conclusion}. 

\section{Mathematical model}
\label{sec:model}

\subsection{The Tero-Kobayashi-Nakagaki model}
\label{sec:tkn}
Let $G$ be an undirected multigraph with node set $N$, edge set $E$, edge lengths $\vec l \in \Real^E_{>0}$ and two distinguished nodes $s_0, s_1 \in N$. The graph models \emph{P. polycephalum}'s vein network; the edges represent the tubular channels, and the nodes represent junctions between the tubes. The two distinguished nodes $s_0$, $s_1$ are two junctions corresponding to the location of the food sources. 

As food is absorbed by the organism, and protoplasmic flow is distributed through the network, the widths of the tubular channels adapt to the flow. In our discussion, $\vec x \in \Real^E_{>0}$ will be a state vector representing the fourth powers of the radii of the tubular channels of the slime mold. For an edge $e \in E$, the value $x_e$ is called the \emph{capacity} of $e$. 
The transport optimization process described in the introduction has previously been modeled \citep{Tero:2007} as a system of coupled, nonlinear ordinary differential equations, 
\begin{equation}
\label{eq:adap} 
\dot{x}_e = f \left(\abs{q_e} \right) - x_e
\qquad \text{ for all } e \in E. 
\end{equation}

Equation \eqref{eq:adap} is called the \emph{adaptation equation}. 
The function $f$ models the response dynamics of the tubular channels to the fluid flow $q_e$ along the edge. Proposed forms include $f(y)=y^\mu$ and $f(y)=(1+\al) y^\mu/(1+\al y^\mu)$, where $\mu, \al >0$ are parameters of the model.  
The dynamic vector $\vec q \in \Real^E$, called the \emph{(fluid) flow}, is determined at any time by the capacity and length of the edges, by solving a network Poisson equation, as follows. Without loss of generality, assume that $N=\{1,2,\ldots,n\}$, $E = \{ 1, 2,\ldots, m \}$ and assume an arbitrary orientation of the edges. Let $\vec B= (B_{ve})_{v \in N, e \in E}$ be the incidence matrix of $G$ under this orientation, that is,
$$B_{ve} \defas \begin{cases} 
+1 & \text{ if } v \text{ is the tail of } e \\
-1 & \text{ if } v \text{ is the head of } e \\
0 & \text{ otherwise}. 
\end{cases}
$$
Then $\vec q$ is defined as the unit-value flow from $s_0$ to $s_1$ of minimum energy, that is, as the unique optimal solution to the following continuous quadratic optimization problem:  
\begin{align}
\label{eq:thomson}
\min\ & \vec q\tp \vec R \vec q \\
\notag \text{s.t. } & \vec B \vec q = \vec b.  
\end{align}
Here, $\vec R \in \Real^{E \times E}$ is the diagonal matrix with value $r_{e} \defas l_{e}/x_{e}$ for the $e$-th element of the main diagonal, and $\vec b \in \Real^N$ is the vector defined by
$$
b_v \defas \begin{cases} 
0 & \text{ if } v \notin \{s_0,s_1\}, \\ 
+1 & \text{ if } v=s_0, \\ 
-1 & \text{ if } v=s_1. 
\end{cases}
$$
It is well-known that a vector $\vec q$ is optimal for system \eqref{eq:thomson} if and only if it satisfies Kirchhoff's circuit laws \citep[Chapter IX]{Bollobas:1998}. In particular, \emph{Kirchhoff's current law} is expressed by the constraint $\vec B \vec q = \vec b$, which, in words, requires that the flow has zero divergence everywhere except at nodes $s_0$ and $s_1$. 
\emph{Kirchhoff's voltage law} is implicit in the optimality condition for \eqref{eq:thomson}, which implies that there exist values $p_1,\ldots,p_n \in \Real$ (the node \emph{potentials}) satisfying the hydrodynamic analogue of \emph{Ohm's law} \citep[Section II.1]{Bollobas:1998}: 
\begin{equation} 
\label{eq:ohm}
q_e = (p_u - p_v)/r_e, \qquad \text{ whenever edge } e \text{ is oriented from } u \text{ to } v.
\end{equation}
The sum of potential differences along any cycle of the network is thus zero.

Node $s_0$ is the \emph{source} of the flow, node $s_1$ the \emph{sink}. It is very important to remark that while the flow has been somewhat arbitrarily directed from $s_0$ to $s_1$, the opposite choice yields exactly the same dynamics, because of the absolute value in \eqref{eq:adap}; the only effect would be to replace $\vec q$ with $-\vec q$. In other words, a flow reversal has no effect on the veins' dynamics. 

The value $r_e$ is called the \emph{resistance} of edge $e$, while $b_v$ is the \emph{divergence} of the flow $\vec q$ at $v$. The constant $b_{s_0}$ (here, $b_{s_0}=1$) is the flow's \emph{value}, that is, the divergence of the flow at the source. 

The quantity $\pdiff \defas \vec q\tp \vec R \vec q$ is the (instantaneous) \emph{energy} of the flow $\vec q$. By the \emph{conservation of energy} principle, the energy of the flow equals the difference between the source and sink potentials, times the value of the flow \citep[Corollary IX.4]{Bollobas:1998}: 
\begin{equation}
\label{eq:conservation}
\pdiff = (p_{s_0}-p_{s_1}) b_{s_0} = p_{s_0}-p_{s_1}. 
\end{equation}  

An alternative way to express the fluid flow vector arises from the Laplacian operator of the graph \citep{Strang:1988,Biggs:1997}. Let $\vec C \defas \vec R^{-1}$. The \emph{Laplacian} of $G$ is the symmetric and positive semidefinite matrix $\vec L \defas \vec B \vec C \vec B\tp$. If we represent the potential vector by $\vec p \in \Real^N$, Ohm's law \eqref{eq:ohm} can be written in matrix form as 
\begin{equation}
\label{eq:ohm-vec}
\vec q = \vec C \vec B\tp \vec p.
\end{equation}
Multiplying both sides by $\vec B$ yields the \emph{discrete Poisson equation} $\vec L \vec p = \vec b$, with solution $\vec p = \vec L^+ \vec b$, where $\vec L^+$ is the Moore-Penrose pseudoinverse\footnote{The \emph{Moore-Penrose pseudoinverse} of a matrix $\vec A$ is the unique matrix $\vec A^+$ such that $\vec A \vec A^+ \vec A= \vec A$, $\vec A^+ \vec A \vec A^+=\vec A^+$, and $\vec A \vec A^+$ and $\vec A^+\vec A$ are both Hermitian.} of $\vec L$.  Substituting in \eqref{eq:ohm-vec}, we get
\begin{equation}
\label{eq:flow-formula}
\vec q = \vec C \vec B\tp \vec L^+ \vec b.
\end{equation}

The fluid flow has been observed to be laminar \citep{Kamiya:1950}. In this case, \emph{Poiseuille's law} expresses the relation between the flow rate, the tube radius, and the pressure gradient: for an edge $e=(u,v)$,  
\begin{equation}
\label{eq:poiseuille}
\abs{q_e} = \frac{\pi R_e^4}{8 \eta} \frac{\abs{p_u-p_v}}{l_e} = \frac{\pi x_e}{8 \eta} \frac{\abs{p_u-p_v}}{l_e},  
\end{equation}
where $R_e\defas x_e^{1/4}$ is the radius of the tube, and $\eta$ is the viscosity constant. 

\subsection{The revised model}
\label{sec:revised-model}
The revised model that we propose resembles closely that discussed in the previous section, following in particular Equations \eqref{eq:thomson}--\eqref{eq:poiseuille}. However, since in this work we conjecture that pressure gradients, rather than sheer flow amounts, control the response of the tubular channels, we take the controlling variables to be the ratios $\abs{p_u-p_v}/l_e$ (instead of the edge flows $\abs{q_e}$). After appropriate normalization, Poiseuille's law implies that the pressure gradients are equivalent to the ratios $\abs{q_e}/{x_e}$. Therefore, we replace Equation \eqref{eq:adap} by 
\begin{equation}
\label{eq:adap2} 
\dot{x}_e = x_e \left( f_e \left(\frac{\abs{q_e}}{x_e} \right) - 1 \right)
\qquad \text{ for all } e \in E. 
\end{equation}
Each edge $e \in E$ has its own dimensionless response function $f_e : \Real_{\ge 0} \to \Real_{\ge 0}$. These response functions are assumed to satisfy the following condition. 
\begin{definition}
\label{def:srf}
A function $f : \Real_{\ge 0} \to \Real_{\ge 0}$ is a \emph{standard response function} if: 
\begin{enumerate}
\item $f(1)=1$; 
\item $f$ is strictly increasing on $\Real_{> 0}$; 
\item $f$ is differentiable on $\Real_{> 0}$. 
\end{enumerate}
\end{definition}
 
Our main result, the stability analysis in Section \ref{sec:parallel}, will not require any other property from the $f_e$, apart from being standard response functions. However, when contrasting our findings with those of \citet{Tero:2007} (Section \ref{sec:comparison}) and in the simulations (Section \ref{sec:simul}), we consider for concreteness the same types of response functions that have been considered in earlier literature: 
\begin{enumerate}
\item[(Type I)] {Nonsaturating response}:  
$f(y) = y^\mu$, for some $\mu>0$; 
\item[(Type II)] {Saturating response}: 
$f(y) = (1+\al) y^\mu / (1+\al y^\mu)$, for some $\mu,\alpha>0$. 
\end{enumerate}
While Type I functions have a simpler structure, Type II functions have the additional property of saturating as $y \to \infty$, implying a finite maximum radius for the tubes, and may therefore be considered more realistic. 
Note in any case that, by using 
\begin{equation}
\label{eq:response}
f(y) = \frac{(1+\al) y^\mu}{1+\al y^\mu}
\end{equation}
with $\mu>0$, $\alpha \ge 0$, one can capture both Type I ($\al=0$) and Type II ($\al>0$) responses. 

We observe that, when using response functions of the form \eqref{eq:response} with $\mu=1$ and $\alpha=0$, the dynamics \eqref{eq:adap2} are identical with \eqref{eq:adap}. We will show, however, that in general they have a qualitatively different behavior (Section \ref{sec:comparison}).  

For the purpose of analysis, we finally assume that the edge length vector $\vec l$ is such that each $s_0$-$s_1$ path in $G$ has a distinct overall length. That is, the configuration of lengths is nondegenerate, which is a physically realistic assumption. While this assumption is not crucial to our main results, it simplifies their statement: for example, without this assumption, the fixed points of the system may not be isolated and the shortest path in the network may not be unique. 

\subsection{Basic properties of the revised model}
We end this section with a couple of simple but useful properties of the revised model. 
\begin{proposition}
\label{prop:pLp}
$\pdiff = \vec p\tp \vec L \vec p$. 
\end{proposition}
\begin{proof}
By \eqref{eq:flow-formula} and the definition of $\vec L$, 
$\pdiff = \vec q\tp \vec R \vec q = \vec b\tp \vec L^+ \vec B \vec C \vec B\tp \vec L^+ \vec b = \vec p\tp \vec L \vec p$. 
\end{proof}

\begin{proposition}
\label{prop:dissipative}
The set $[0,1]^E$ is an attracting set for the dynamics \eqref{eq:adap2}. 
\end{proposition}
\begin{proof}
Observe that, for any edge $e\in E$, $0 \le \abs{q_e} \le 1$, since the flow on any edge cannot exceed the global flow value, which is $b_{s_0}=1$. Therefore, as long as $x_e > 1$, we have $\abs{q_e}/x_e < 1$ and $\dot x_e < x_e \cdot (f_e(1) - 1) = 0$. 
\end{proof}

\section{Stability analysis for parallel-edge networks}
\label{sec:parallel}
\subsection{Network structure}
To allow analytical tractability, in this section we limit our discussion to simple networks consisting of two nodes and a set of parallel edges between them (Figure \ref{fig:paredges}(a)); a special case is the ring-shaped network (Figure \ref{fig:paredges}(b)). 
While admittedly a simplification, parallel-edge and ring-shaped networks already exhibit a wide range of dynamical properties and have been the departure point of previous analyses; see, for example, the discussion by \citet{Tero:2007}. 

\begin{figure}
\centering
\subfigure[A parallel-edge network]
{\includegraphics[scale=0.5]{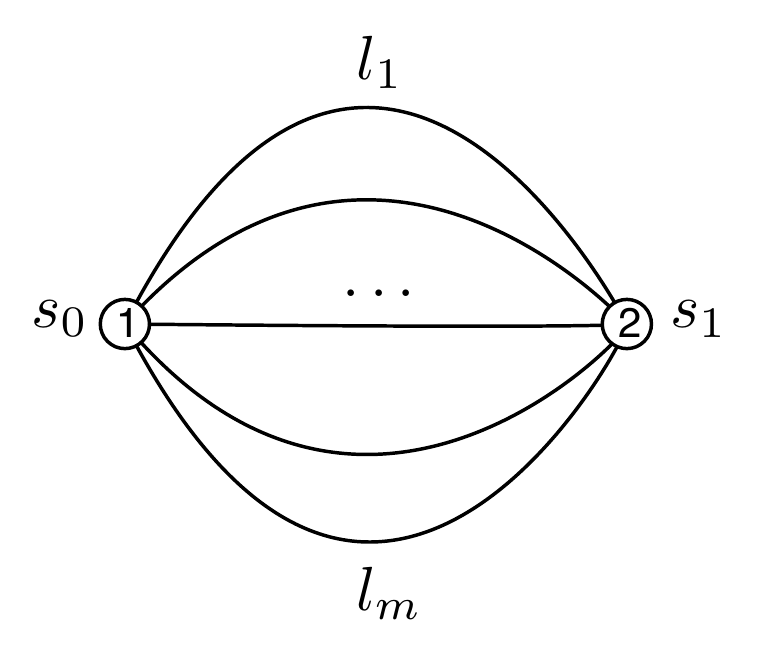}}
\subfigure[A ring-shaped network]
{\includegraphics[scale=0.5]{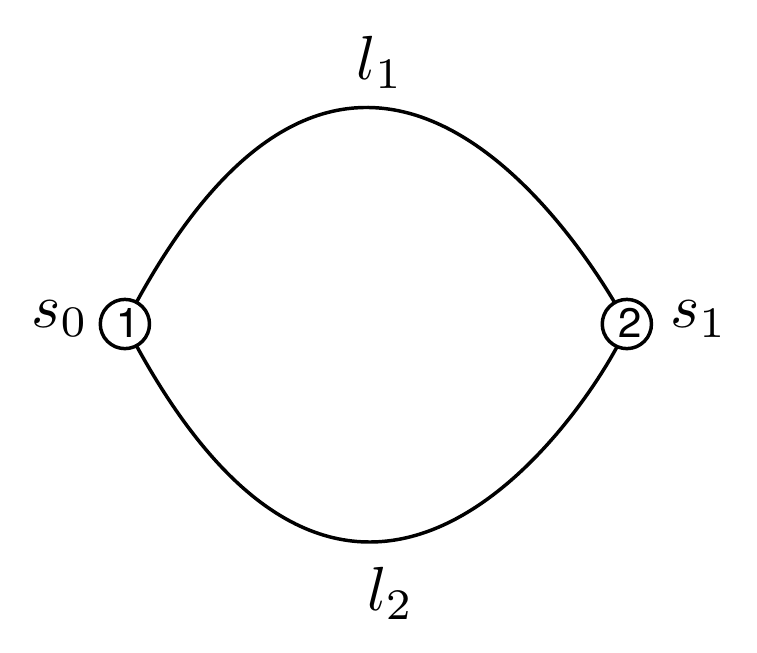}}
\caption{Examples of parallel-edge network topologies. Labels inside the nodes denote their identifier. Labels $s_0$ and $s_1$ denote the source and sink node, respectively. The label on each edge describes its length.}
\label{fig:paredges} 
\end{figure}

The incidence matrix $\vec B$ of a parallel-edge network has the simple structure
\begin{equation}
\label{eq:parlinks-structure}
\vec B = \left( \begin{array}{cccc}
+1 & +1 & \cdots & +1 \\
-1 & -1 & \cdots & -1 
\end{array} \right) .
\end{equation}
Using this structure, the instantaneous energy of the system is easily derived. 
\begin{proposition}
\label{prop:P}
In a parallel-edge network $G=(N,E)$,
\begin{equation}
\label{eq:pot}
\pdiff = (\trace \vec C)^{-1} = \left(\sum_{e\in E} x_e/l_e \right)^{-1}.  
\end{equation}
\end{proposition}
\begin{proof}
The matrix $\vec B$ given by \eqref{eq:parlinks-structure} yields the Laplacian
$$ \vec L= \vec B \vec C \vec B\tp = \left(\begin{array}{cc}
\trace \vec C & -\trace \vec C \\
-\trace \vec C & \trace \vec C. 
\end{array}
\right) 
$$
and combining \eqref{eq:conservation} and Proposition \ref{prop:pLp} yields
$$ \pdiff  = \vec p\tp \vec L \vec p = (p_1-p_2)^2 \trace \vec C = \pdiff^2 \trace \vec C. $$
Solving for $\pdiff$ yields the claim. 
\end{proof}

\subsection{Location of fixed points}
A \emph{fixed point} of \eqref{eq:adap2} is a vector $\vec x \in \Real^E_{\ge 0}$ such that $$x_e \cdot \left( f_e \left(\frac{\abs{q_e}}{x_e} \right) - 1 \right) = 0 \text{ for all } e \in E.$$  

\begin{lemma}
\label{lem:basis}
The fixed points of \eqref{eq:adap2} in a parallel-edge network $G=(N,E)$ are exactly the standard basis vectors of $\Real^E$. 
\end{lemma}
\begin{proof}
Since the response function $f_e$ is assumed to be standard (Definition \ref{def:srf}), the unique solution to $f(y)=1$ is $y=1$. The fixed point condition for \eqref{eq:adap2} is thus equivalent to 
$$ (x_e = 0) \text{ or } (\abs{q_e}/x_e = 1) \qquad \text{ for all } e \in E. $$
In a parallel-edge network, each term $\abs{q_e}/x_e$  simplifies to $\pdiff/l_e$ by Ohm's law \eqref{eq:ohm} and by \eqref{eq:conservation}. Moreover, we assumed that no two source-sink paths have the same length. This implies that, in a fixed point, there cannot be two distinct $x_e$'s with $x_e \neq 0$; which in turn implies that in a fixed point, $x_i=q_i=1$ for exactly one $i \in E$ and $x_{e}=q_e=0$ for all $e \neq i$. Conversely, it is straightforward to verify that any standard basis vector $\vec \chi_i$ of $\Real^E$ is a fixed point of \eqref{eq:adap2}. 
\end{proof}

\subsection{Nature of fixed points}
Recall that our assumption on the response functions $f_e: \Real_{\ge 0} \to \Real_{\ge 0}$ is that they are increasing, differentiable, and such that $f_e(1)=1$. 
After substitution in \eqref{eq:adap2}, using $\abs{q_e}/x_e = \pdiff/l_e$,  the  adaptation equation can be equivalently written as
\begin{equation}
\label{eq:adap3}
\dot{x}_e = x_e \left( f_e \left( \frac{\pdiff}{l_e} \right) - 1 \right) \qquad \text{ for all } e \in E, 
\end{equation}
where, as observed in Proposition \ref{prop:P}, $\pdiff=(\sum_e x_e/l_e)^{-1}$. 

\begin{theorem}
\label{thm:stability}
System \eqref{eq:adap3} has exactly one stable fixed point, namely, the standard basis vector $\vec \chi_{i^*}$ corresponding to the edge $i^*$ of shortest length in the network. All other fixed points are unstable. 
\end{theorem}
\begin{proof}
Let $\vec x \in \Real^E_{\ge 0}$ and $i, j \in E$ ($i\neq j$). 
Direct computation of the terms of the Jacobian matrix $\vec J(\vec x)$ of \eqref{eq:adap3} yields, using the substitution $y=\pdiff/l_i$,  
\begin{align*}
J_{ii}(\vec x) &= f_i \left( \frac{\pdiff}{l_i} \right) - 1 + x_i \,\frac{\partial f_i}{\partial x_i} \left(\frac{\pdiff}{l_i} \right)
\\ 
&= f_i \left( \frac{\pdiff}{l_i} \right) - 1 + x_i \, \frac{\partial y}{\partial x_i} \frac{\partial f_i}{\partial y} \left(\frac{\pdiff}{l_i}\right) \\
&= f_i \left( \frac{\pdiff}{l_i} \right) - 1 - x_i \, \frac{1}{l_i^2} \left(\sum_e \frac{x_e}{l_e}\right)^{-2} f'_i\left( \frac{\pdiff}{l_i} \right) \\
&= f_i \left( \frac{\pdiff}{l_i} \right) - 1 - x_i \, \frac{\pdiff^2}{l_i^2} f'_i\left( \frac{\pdiff}{l_i} \right), \\
J_{ij}(\vec x) &= x_i \, \frac{\partial f_i}{\partial x_j} \left(\frac{\pdiff}{l_i}\right) \\
&= x_i \, \frac{\partial y}{\partial x_j} \frac{\partial f_i}{\partial y} \left(\frac{\pdiff}{l_i}\right) \\
&= - x_i \, \frac{\pdiff^2}{l_i l_j} f'_i\left(\frac{\pdiff}{l_i}\right). 
\end{align*}
We evaluate the Jacobian at any standard basis vector $\vec \chi_i$ to obtain (for distinct $i,j,k \in E$)
\begin{align*}
J_{ii}(\vec\chi_i) &= -f'_i(1)  \\
J_{ij}(\vec\chi_i) &= -(l_i/l_j) f'_i(1) \\
J_{ji}(\vec\chi_i) &= 0 \\
J_{jj}(\vec\chi_i) &= f_j(l_i/l_j) - 1 \\
J_{jk}(\vec\chi_i) &= 0 \\
J_{kj}(\vec\chi_i) &= 0. 
\end{align*}
After rearranging the rows and columns of $\vec J(\vec\chi_i)$ so that the $i$th column and row are swapped with the first column and row, respectively, we obtain the following matrix $\vec J^{(i)}$: 
$$
\vec J^{(i)} = \left(\begin{array}{cccc}
-f'_i(1) & -(l_i/l_2) f'_i(1) & \ldots & -(l_i/l_m) f'_i(1) \\
0 & f_2(l_i/l_2)-1 & \ldots & 0 \\
0 & 0 & \ddots & 0 \\
0 & 0 & \ldots & f_m(l_i/l_m)-1 
\end{array}\right)
$$
where, for $j \notin \{1,i\}$, the $j$th element on the main diagonal is $f_j(l_i/l_j)-1$, and for $j=i$ it is $f_1(l_i/l_1)-1$. 
By construction, $\vec J^{(i)}$ and $\vec J(\vec\chi_i)$ have the same eigenvalues; the advantage of $\vec J^{(i)}$ is that it is upper triangular, and so its eigenvalues can be read off its main diagonal. These eigenvalues are 
$$ -f'_i(1), \, f_2(l_i/l_2)-1, \, \ldots, f_1(l_i/l_1)-1, \ldots, f_m(l_i/l_m)-1. $$
By our assumptions on the response functions, it holds that $-f'_i(1)<0$ and the sign of $f_j(l_i/l_j)-1$ is the same as the sign of $l_i-l_j$. We conclude that the eigenvalues associated to an equilibrium point $\vec\chi_i$ are all negative if and only if $l_i < l_j$ for all $j \neq i$. Otherwise, at least one eigenvalue is positive, and the fixed point is a source or a saddle. Consequently, there is exactly one stable fixed point of the dynamics \eqref{eq:adap3}, corresponding to the edge with shortest length; all other fixed points are unstable. We remark that this conclusion holds independently of the concrete form of the response functions, as long as they satisfy Definition \ref{def:srf}. 
\end{proof}

Note that Theorem \ref{thm:stability} does not rule out the existence of periodic orbits in phase space. To exclude this possibility, we show convergence to equilibrium for all trajectories. 

\begin{lemma}
\label{lem:converge}
Every trajectory of \eqref{eq:adap3} converges to an equilibrium as $t \to \infty$. 
\end{lemma}
\begin{proof}
We claim that the function
$$ V(\vec x) \defas \vec 1\tp \vec x + \ln \pdiff = \sum_{e \in E} x_e - \ln \left(\sum_{e \in E} x_e/l_e \right) $$
is a Lyapunov function for \eqref{eq:adap3}. First observe that $V$ is bounded from below, as all $x_e \ge 0$ and $\pdiff$ is also bounded from below (by Proposition \ref{prop:dissipative}, we can assume that $x_e \le 1+\eps$ for each $e \in E$, after some finite time). 
Then we compute 
\begin{align*}
\dot{V}(\vec x) &= \sum_{e \in E} \frac{\partial V}{\partial x_e} \dot x_e \\
&= \sum_{e \in E} \left(1 - \frac{\pdiff}{l_e} \right) x_e \left( f_e \left(\frac{\pdiff}{l_e} \right) - 1 \right) \\
&= - \sum_{e \in E} x_e \left( \frac{\pdiff}{l_e} - 1 \right) \left( f_e \left(\frac{\pdiff}{l_e} \right) - 1 \right) \\
& \le 0, 
\end{align*}
where the last inequality follows from the facts that $x_e \ge 0$ for each $e \in E$, and that $\pdiff/l_e - 1$ has the same sign as $f_e(\pdiff/l_e) - 1$ (since $f_e$ is increasing and $f_e(1)=1$). 
Moreover, from the same derivation it follows that $\dot V(\vec x) = 0$ if and only if, for all $e \in E$, either $x_e = 0$ or $\pdiff = l_e$, that is, if and only if $\vec x$ is a fixed point of \eqref{eq:adap3}. 
\end{proof}

Given Theorem \ref{thm:stability} and Lemma \ref{lem:converge}, it is to be expected that all trajectories starting in the interior of the positive orthant converge to the stable fixed point. This is indeed formalized in our final result. 

\begin{theorem}
Let $\vec x(0)> \vec 0$. As $t \to \infty$, $\pdiff(t)$ converges to $l_{i^*}$, where $i^*$ is the shortest edge in $E$, and $\vec x(t)$ converges to $\vec \chi_{i^*}$. 
\end{theorem}
\begin{proof}
By Lemma \ref{lem:converge}, it suffices to prove the first part of the claim. Suppose by contradiction that $\pdiff$ converges to the length of some other edge $e \neq i^*$. Let $\delta \defas (l_e - l_{i^*}) / 2 l_{i^*} > 0$ and define $$W(t) \defas \ln x_{i^*}(t). $$ Then for all sufficiently large $t$, we may assume $\pdiff(t) \ge l_e - \delta l_{i^*}= (1 + \delta) l_{i^*}$, and therefore
\begin{align*}
\dot{W} &= \frac{\dot x_{i^*}}{x_{i^*}}   \\
&= \frac{x_{i^*}}{x_{i^*}} \left( f_{i^*} \left( \frac{\pdiff}{l_{i^*}} \right) - 1 \right)  \\
&= f_{i^*} \left( \frac{\pdiff}{l_{i^*}} \right) - 1  \\
&\ge f_{i^*} \left( 1 + \delta \right) - 1, 
\end{align*}
so that $\dot W(t)$ is larger than some positive constant for all sufficiently large $t$. This implies $W(t) \to \infty$. On the other hand, Proposition \ref{prop:dissipative} implies that for any fixed $\eps>0$ and all large $t$, $x_{i^*}(t) \le 1+\eps$, so that $W(t) \le \ln (1+\eps)$, yielding a contradiction. 
\end{proof}

\begin{figure}
\centering
\includegraphics[scale=0.4]{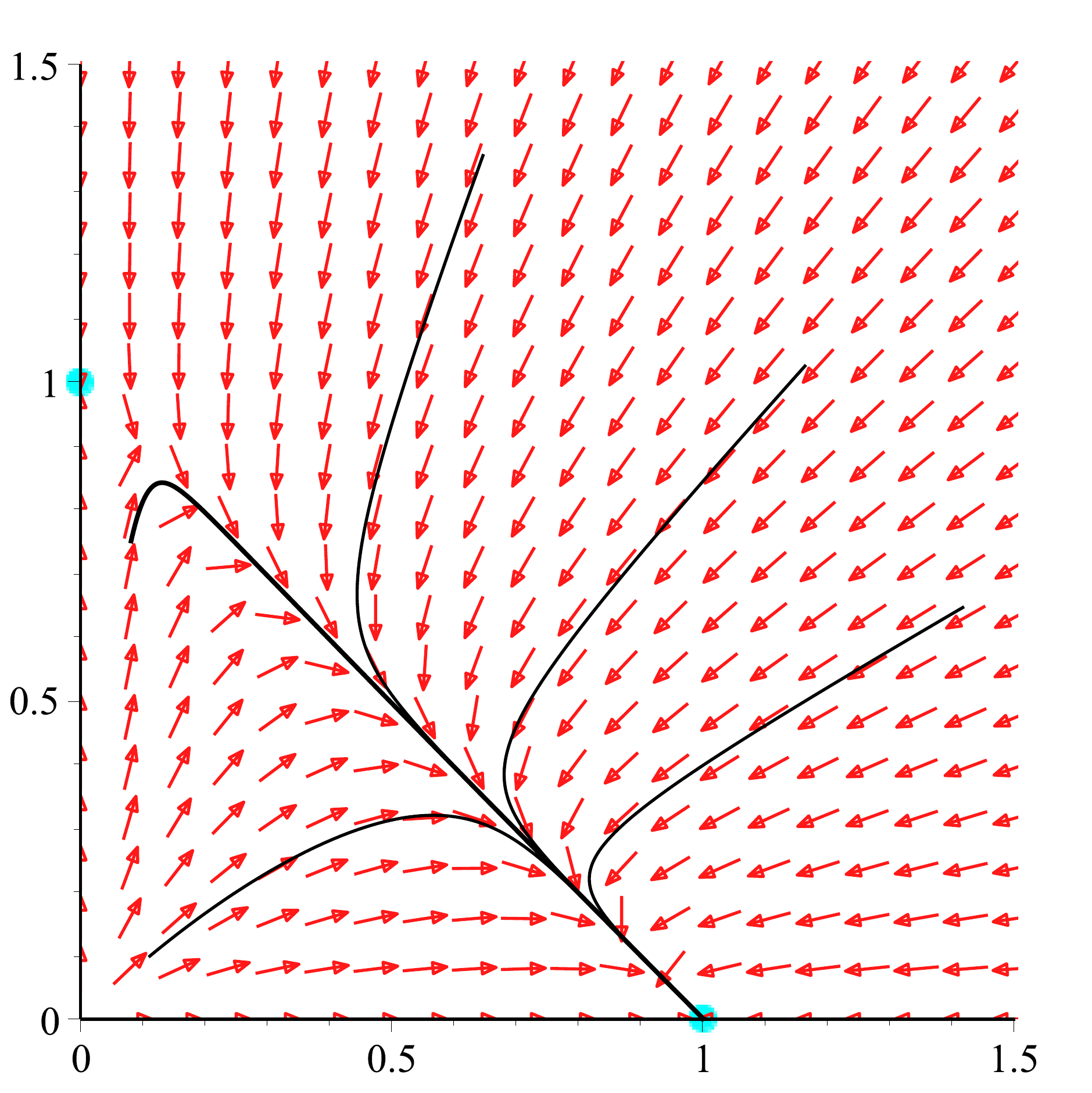}
\caption{A typical phase portrait of \eqref{eq:adap3} for a ring-shaped topology with a Type II response. Parameter values are $l_1=1$, $l_2=1.2$, $\mu=0.8$, $\alpha=1$. Several solutions are shown, all converging to the shortest-path equilibrium $\vec \chi_1=(1,0)$.}
\label{fig:phase}
\end{figure}
Figure \ref{fig:phase} shows a typical phase portrait of \eqref{eq:adap3} for the ring-shaped topology. 

\subsection{Comparison with existing models}
\label{sec:comparison}
It is useful to contrast our findings with those obtained in the original flow-based model by \citet{Tero:2007} in the case of a ring-shaped network. 
The flow-based model of Tero et al.~undergoes a bifurcation for both Type I and Type II response functions around the value $\mu=1$. When $\mu>1$, there are two stable equilibria and one unstable equilibrium; when $\mu<1$, there is a single stable equilibrium (and two unstable ones), but it lies in the interior of the positive orthant and therefore it does not correspond to any edge of the graph. Only when $\mu=1$ the equilibria are two, one for each edge, and the stability of each equilibrium depends on the length of the corresponding edge. 

In contrast, in the pressure-gradient based model proposed in this article, in a ring-shaped network there are always two equilibria, one of which is stable and the other one of which is not, independently of the details of the response function. The stable equilibrium is always the one corresponding to the shortest path in the network.

\section{Simulation of other network topologies}
\label{sec:simul}
In this section we simulate the dynamics \eqref{eq:adap2} on two more general network topologies, to which the analysis of Section \ref{sec:parallel} does not apply. 

We consider the bridge-shaped Wheatstone network of Figure \ref{fig:topo}(a), as well as the more complex network shown in Figure \ref{fig:topo}(b). The choice of these topologies is motivated by the fact that they are inherently non-series-parallel networks. In both figures, the shortest path is highlighted in bold; the length of the shortest path in the two networks is 2.5 and 90, respectively. 

\begin{figure}
\centering
\subfigure[Wheatstone network] 
{\includegraphics[scale=0.5]{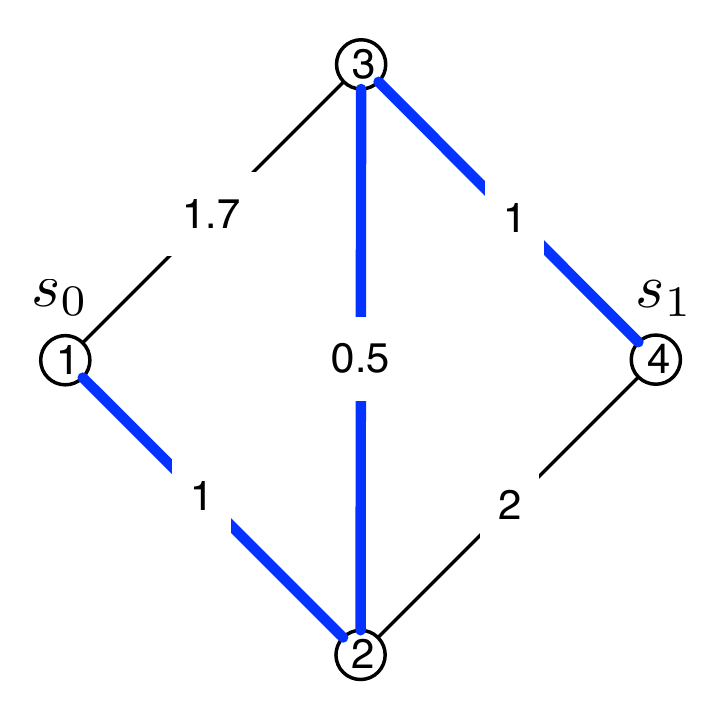}}
\hspace{5mm}
\subfigure[Another arbitrary topology] 
{\includegraphics[scale=0.5]{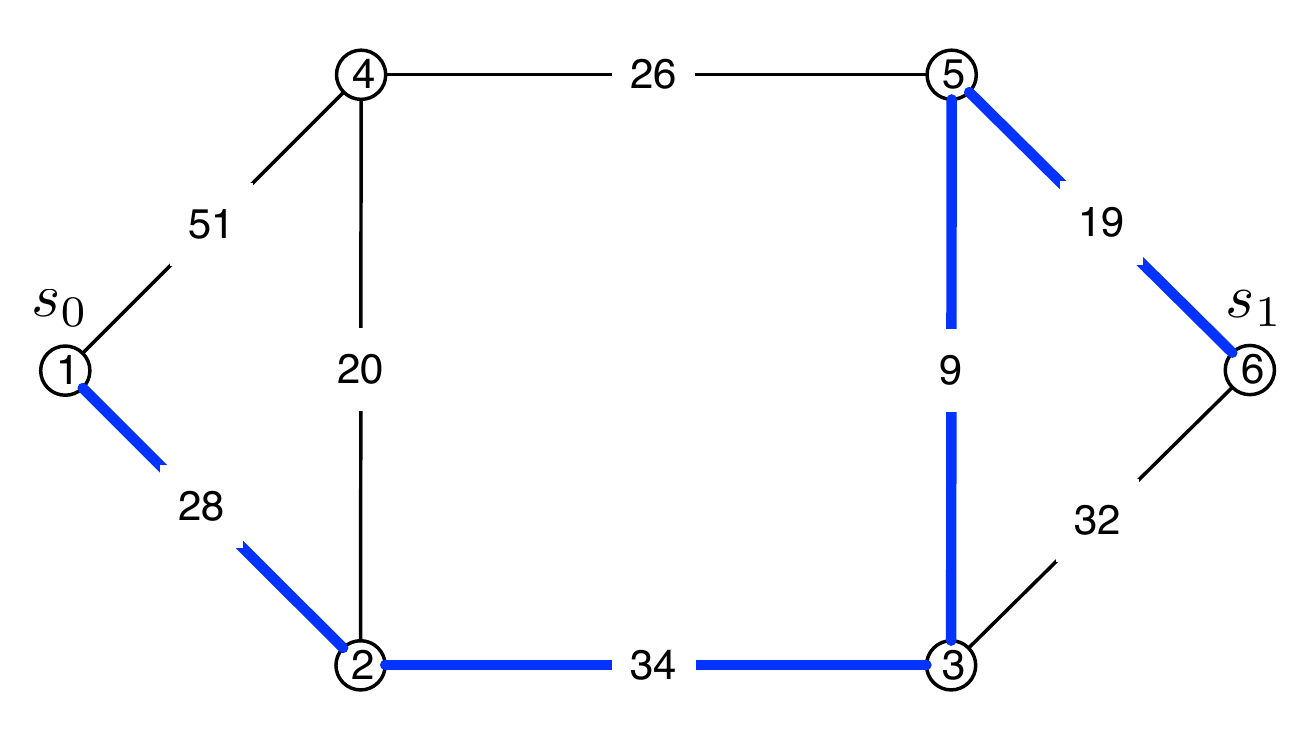}}
\caption{Examples of non-parallel-edge network topologies. Labels inside the nodes denote their identifier. Labels $s_0$ and $s_1$ denote the source and sink node, respectively. The label on each edge describes its length. The bold edges identify the shortest source-sink path.}
\label{fig:topo}
\end{figure}

We simulate the dynamics \eqref{eq:adap2} under an Euler discretization scheme with a stepsize $h = 0.1$: 
\begin{align}
\label{eq:discrete}
x_e[t+1] - x_e[t] &= h \cdot x_e[t] \left( f_e \left( \frac{q_e[t]}{x_e[t]} \right) - 1 \right) \text{ for all } e \in E, \\
\notag
 \vec q[t] &= \vec C[t] \vec B\tp \vec p[t], \\
\notag
\vec p[t] &= {(\vec B \vec C[t] \vec B\tp)}^+ \vec b, \\
\notag
\vec C[t] &= \mathrm{diag}(x_1[t]/l_1, \ldots, x_m[t]/l_m), 
\end{align}
where, as in Section \ref{sec:model}, $(\vec B \vec C[t] \vec B\tp)^+$ is the pseudoinverse of the Laplacian $\vec L[t] = \vec B \vec C[t] \vec B\tp$. 

As the $f_e$, we assume identical response functions of the form \eqref{eq:response}. As the initial condition, we select the symmetric state $\vec x[0] = \vec 1$, to ensure absence of bias towards any specific fixed point. We note that, in analogy with the case of a parallel-edge network, each $s_0$-$s_1$ path in the network is in one-to-one correspondence with a fixed point of \eqref{eq:discrete}. 

Since we expect each $x_e$ to approach either 0 or 1 as $t \to \infty$, we expect $\vec l\tp \vec x[t]$ to approach the length of some source-sink path in the network. Therefore, we define the quantity $\vec l\tp \vec x[t]$ to be the \emph{transport cost} at time $t$. The dynamics have an optimal behavior if, in the limit of large $t$, the transport cost approaches the length of the shortest path. 

Figures \ref{fig:wheat_plot}(a) and \ref{fig:g6_ter_plot}(b) plot the value of the transport cost for $t \in \{1,10,1000\}$ for the two networks, as a function of the power-law exponent $\mu$ and the saturation parameter $\alpha$, with ranges $\mu \in [0.5,1.5]$, $\alpha \in [0,2]$. The data confirm that indeed, as time passes, the transport cost approaches the optimal costs given by the shortest path lengths, independently of the values of the parameters $\mu$ and $\alpha$. However, the convergence speed is affected by the parameters; namely, the dynamics appear to be faster when the power-law exponent $\mu$ is large and the saturation parameter $\alpha$ is small. 

\begin{figure}[ht!]
\centering
\subfigure[Transport cost values for the network of Figure \ref{fig:topo}(a) after 1 (top surface), 10 (middle), 1000 (bottom) Euler steps, as a function of $\mu$ and $\alpha$.]
{\includegraphics[scale=.8]{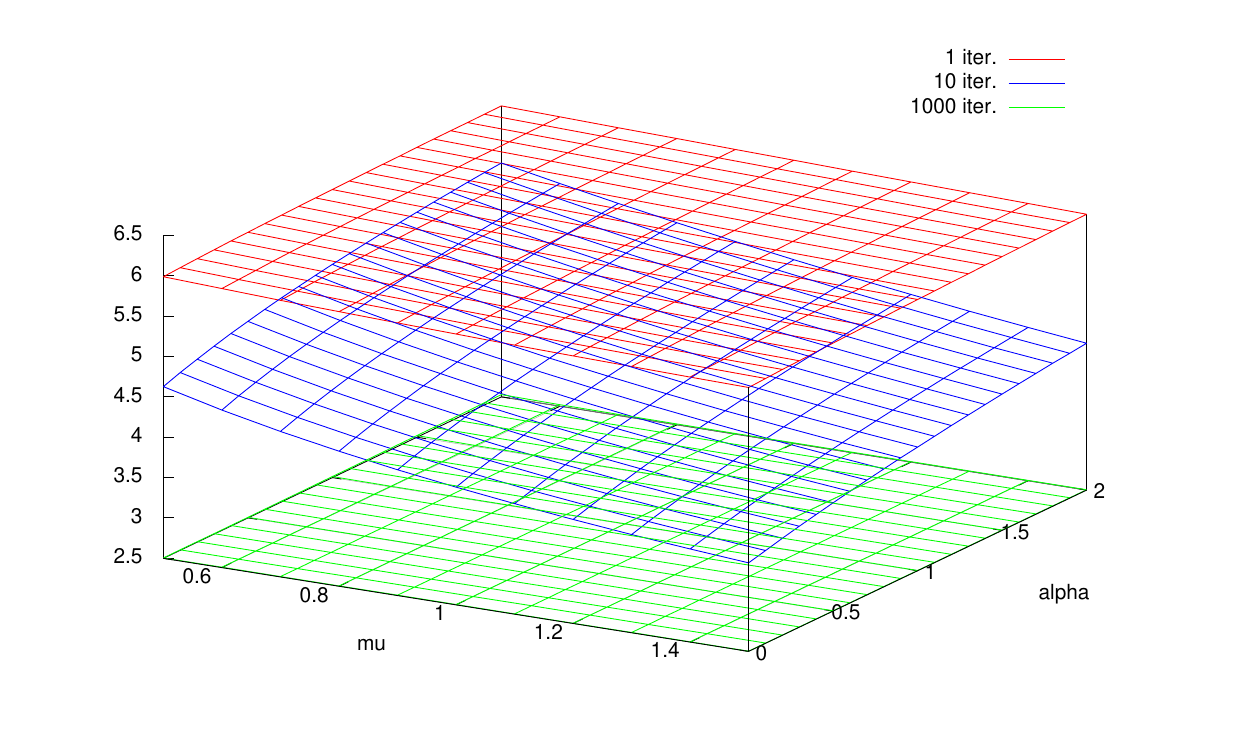}}
\label{fig:wheat_plot}
\\ 
\subfigure[Transport cost values for the network of Figure \ref{fig:topo}(b) after 1 (top surface), 10 (middle), 1000 (bottom) Euler steps, as a function of $\mu$ and $\alpha$.]
{\includegraphics[scale=.8]{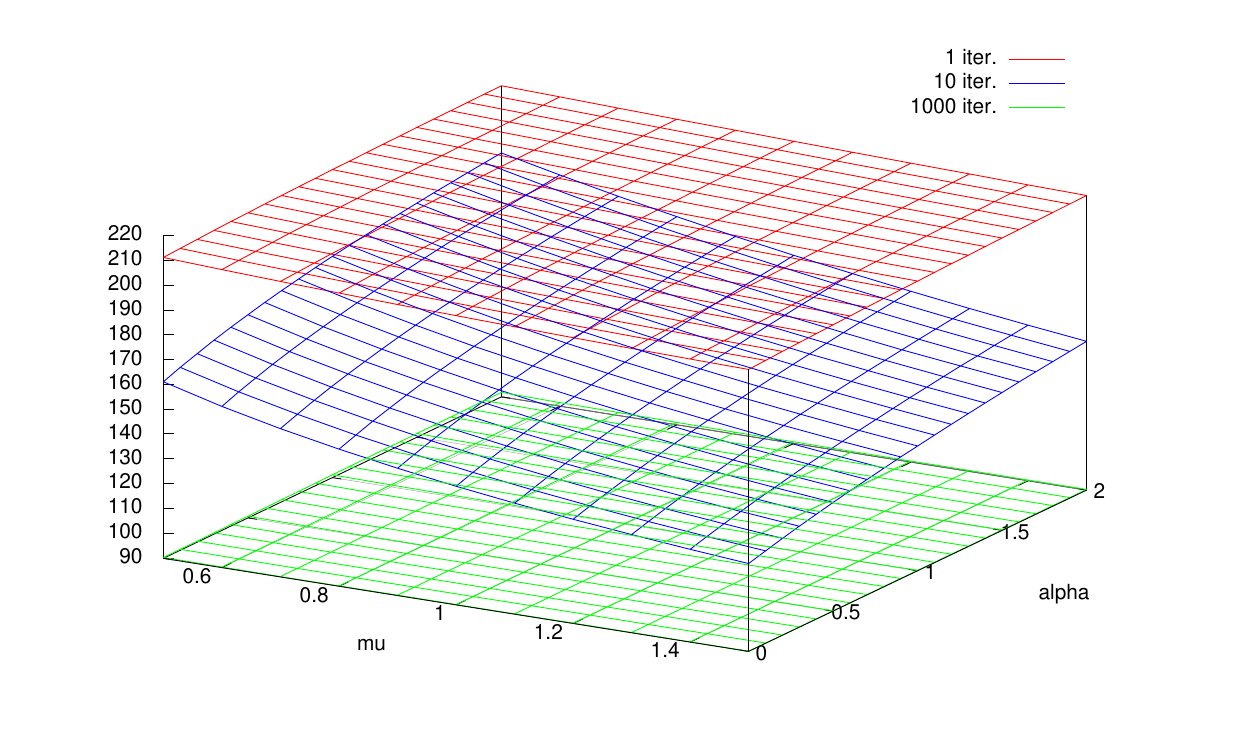}}
\label{fig:g6_ter_plot}
\caption{Transport cost values}
\end{figure}

\section{Discussion and concluding remarks}
\label{sec:conclusion}
The optimization of transport networks is a commonly occurring feature of several natural (as well as artificial) systems: blood vasculature and leaf venation are two examples. 
The fluid transport optimizing behavior of \emph{P.~polycephalum} may not be surprising in light of the idea that a more efficient use of the available resources (the size of the tubular structures) enables the organism to achieve a higher fitness. However, the accuracy achieved by the positive feedback mechanism between the pressure gradients along the veins and the widths of the tubular channels is somewhat remarkable: the steady state solution is not only approximately or locally optimal; at least in the case of a parallel-edge network, it is the globally optimal solution from the point of view of the total length of the tubes. 

In previous models, based on sheer amounts of flow, this global optimization behavior was known to rely on very specific values of the power-law exponent and of the saturation parameter ($\mu=1$ and $\alpha=0$). We have shown that a model where the controlling variables are the pressure gradients is, instead, able to support the global optimization behavior of the dynamics for a much wider class of response functions or range of parameters. It is a natural, though perhaps formidable, open problem to prove this conjecture analytically for network topologies that go beyond the simple parallel-edge topology we considered in this article. 

\paragraph{Acknowledgements}
The author would like to thank Kurt Mehlhorn and two anonymous reviewers for suggesting several improvements, as well as Alberto Gandolfi and Carmela Sinisgalli for fruitful discussions on \emph{P.~polycephalum}'s dynamics. 

\bibliographystyle{abbrvnat}      

\end{document}